\documentclass[11pt, oneside, english]{article}

\usepackage[english]{babel}
\usepackage[utf8x]{inputenc}
\usepackage[T1]{fontenc}
\usepackage[letterpaper,top=1in,bottom=1in,left=1in,right=1in,marginparwidth=1.75cm]{geometry}
\usepackage{dsfont}
\usepackage{amsmath, amssymb, amsthm, verbatim,enumerate,bbm}
\usepackage{indentfirst, bbm}
\usepackage{appendix}
\usepackage{enumerate}
\usepackage{tikz}
\usepackage{verbatim}
\usepackage[colorinlistoftodos]{todonotes}
\usepackage[colorlinks=true, allcolors=blue]{hyperref}
\usepackage[nameinlink,capitalise,noabbrev]{cleveref}
\usepackage{subfig}
\crefname{appsec}{Appendix}{Appendices}
\creflabelformat{enumi}{#2textup{#1}#3}

\allowdisplaybreaks
\makeatletter
\newtheorem{theorem}{Theorem}[section]

\newtheorem*{namedtheorem}{\theoremname}
\newcommand{\theoremname}{testing}

\newtheorem{lemma}[theorem]{Lemma}

\newtheorem{proposition}[theorem]{Proposition}

\newtheorem*{question*}{Question}

\theoremstyle{definition}

\newtheorem{remark}[theorem]{Remark}

\theoremstyle{plain}

\makeatother

\title{A Concentration Inequality for the Facility Location Problem}
\author{
Sandeep Silwal \\
MIT \thanks{77 Massachusetts Avenue
Cambridge, MA 02139.
\tt{silwal@mit.edu}}
}
\date{}

\begin{document}
\maketitle
\global\long\def\R{\mathbb{R}}

\global\long\def\S{\mathbb{S}}

\global\long\def\Z{\mathbb{Z}}

\global\long\def\C{\mathbb{C}}

\global\long\def\Q{\mathbb{Q}}

%\global\long\def\N{\mathbb{N}}

\global\long\def\P{\mathbb{P}}

\global\long\def\F{\mathbb{F}}

\global\long\def\U{\mathcal{U}}

\global\long\def\V{\mathcal{V}}

\global\long\def\E{\mathbb{E}}

\global\long\def\Ev{\mathscr{Rk}}

\global\long\def\Dg{\mathscr{D}}

\global\long\def\Ndg{\mathscr{ND}}

\global\long\def\Rv{\mathcal{R}}

\global\long\def\Gv{\mathscr{Null}}

\global\long\def\Hv{\mathscr{Orth}}

\global\long\def\Supp{{\bf Supp}}

\global\long\def\Sv{\mathscr{Spt}}

\global\long\def\ring{\mathfrak{R}}

\global\long\def\Bad{{\boldsymbol{B}}}

\global\long\def\supp{{\bf supp}}

\global\long\def\A{\mathcal{A}}

\global\long\def\L{\mathcal{L}}

\newcommand{\event}{\mathcal E}

\newcommand{\p}{\mathbb P}

\newcommand{\eps}{\varepsilon}

\newcommand{\1}{\boldsymbol 1}
\newcommand{\N}{\mathcal{N}}

\begin{abstract}
We give a concentration inequality for a stochastic version of the facility location problem. We show the objective $C_n = \min_{F \subseteq [0,1]^2}|F|+\sum_{x\in X}\min_{f\in F}\|x-f\|$ is concentrated in an interval of length $O(n^{1/6})$ and $\E[C_n]=\Theta(n^{2/3})$ if the input $X$ consists of i.i.d. uniform points in the unit square. Our main tool is to use a geometric quantity, previously used in the design of approximation algorithms for the facility location problem, to analyze a martingale process. Many of our techniques generalize to other settings.
\\

\noindent Keywords: Facility Location; Concentration Inequality; Stochastic Optimization
\end{abstract}

\section{Introduction}
Let $X$ be a set of $n$ points in $D \subset \R^d$. The (minimum) facility location problem (with uniform demands) is the problem of finding a set of points $F \subset D$ (called facilities or centers)  to minimize the objective
\begin{equation} \label{eq:objective}
C_n(X) = \min_{F \subset D} \, |F| + \sum_{x\in X} \min_{f \in F} \| x-f\|.
\end{equation}

The facility location problem is a well studied combinatorial optimization problem and is NP-hard in general. As is the case of many other NP-hard combinatorial optimization problems, stochastic versions of these problems have been studied (see \cite{random_ref1,random_ref2, random_ref3, random_ref4} and the book \cite{book} for examples in TSP, MST, and many other problems). In this paper, we study the stochastic version of the facility location problem. In particular, if our domain $D$ is the unit square $[0,1]^2$ in $\R^2$, our result presented in Theorem \ref{thm:strong} states that $C_n$ is concentrated in an interval of length $O(n^{1/6})$ and satisfies the following concentration bound
\[ \Pr(|C_n -\E[C_n] | \ge tn^{1/6}) \le \exp(-ct^2) \]
where $\E[C_n] = \Theta(n^{2/3})$. However, our techniques are more general and can be extended to other domains and distributional assumptions.

To give more context to our result, we compare our bound against Rhee and Talagrand's concentration result for the $k$-median problem \cite{median}. The $k$-median problem is a related optimization problem where only the second term of the objective in \eqref{eq:objective} appears and where we are constrained to $|F| = k$. Rhee and Talagrand showed in \cite{median} that the cost of the objective function for the $k$-median problem concentrates on an interval of length $O(\sqrt{n/k})$. This follows from the following theorem.

\begin{theorem}[Theorem B in \cite{median}]
Let $Q_n$ be the random variable which denotes the cost of the $k$-median problem on the unit square in $\R^2$ where $n$ points are drawn independently and uniformly at random from the unit square. There exists a constant $c > 0$ such that 
\[\Pr(|Q_n - \E[Q_n]| \ge t) \le \exp(-ct^2k/n). \]
\end{theorem}

While their techniques aren't applicable in our setting, we can interpret our results as `plugging in a specific value' of $k = n^{2/3}$ even though $|F|$ is a random variable in our case. 

Our proof strategy relies on standard martingale tools but uses a more geometric and `local' representation of $C_n$ that allows us to better track the objective cost as new random points are drawn. This geometric formulation is stated in Section \ref{sec:prelim} and has been previously used in algorithmic works related to the facility location problem \cite{MP_alg, sublin_MP}. We also present a weaker concentration result using Talagrand's concentration inequality in Theorem \ref{thm:weak} which we conjecture gives us the optimal concentration result for a variety of settings. We leave it as an interesting open problem to verify this conjecture.

Lastly, we note that while many of our techniques can be adapted to more general distributions and domains, we mainly stick to the uniform distribution on the unit square in $\R^2$ for our presentation due to simplicity and clarity since this case already conveys our ideas. Furthermore, the uniform distribution is the most well-studied case for stochastic combinatorial optimization problems in general and has lead to a wide array of results and influential techniques; for example, in matchings \cite{matching1, matching2, matching3, matching4}, minimum-spanning trees \cite{mst1, mst2, mst3, mst4, mst5, mst6, mst7, mst8}, the traveling salesman problem \cite{tsp1, tsp2, tsp3, tsp4, tsp5}, bin packings \cite{bin1, bin2, bin3, bin4, bin5}, $k$-SAT \cite{sat1, sat2, sat3, sat4, sat5}, and many more problems. See the references within the cited papers, the book \cite{book} and the excellent set of notes in \cite{notes} for further examples.

\subsection{Related Work}
Piersma considered a different formulation of stochastic facility location \cite{related}. In their work, they consider a capacitated version of facility location where each facility is only allowed to `serve' a fixed number of points. Their formulation is given by an integer program with randomly drawn coefficients for their linear constraints. In contrast, our input points are random and the cost to connect a point to a facility is given by Euclidean distances rather than randomly drawn values. This leads to their integer program having a non zero probability of being infeasible whereas in our setting it is always possible to find a solution.

In addition, the `scaling' of the cost for our formulation is naturally on the order of $n^{2/3}$ whereas in \cite{related}, the scaling is $n$. Furthermore, the goal of \cite{related} is to mostly study the convergence of the cost of their formulation using central limit type theorems whereas for us we are more concerned with concentration. Finally, our formulation is more geometric and closely related to the stochastic $k$-median problem studied previously in \cite{median}.

\section{Preliminaries}\label{sec:prelim}
Our points are given by $X = (X_1, \cdots, X_n)$ and all our asymptotics are as $n \rightarrow \infty$. 

 A key point about Rhee and Talagrand's method is that it relies heavily on the fact that the $k$-median objective is only composed of `local' terms (representing the cost incurred by every input point). Local behaviour is often the key to getting concentration bounds for complicated processes, even in other settings such as random graphs, since they allow tools such as bounded difference or Lipschitz concentration inequalities to be used. On first glance, the facility location problem does not seem to have nice local structures due to the additional global $|F|$ term. However in the algorithmic literature about facility location, a suitable local geometric quantity has been considered which will form the basis of our analysis. This is additionally interesting since it's an example of an algorithmic result being used in probability. The geometric quantity is defined as follows.

Let $B(p,r)$ denote the ball of radius $r$ centered at $p$. For each $p \in X$, define radius $r_p > 0$ to satisfy the following relation. 
\begin{equation}\label{eq:rdef}
    \sum_{q \in B(p, r_p) \cap X} (r_p - \|p-q\|) = 1.
\end{equation}

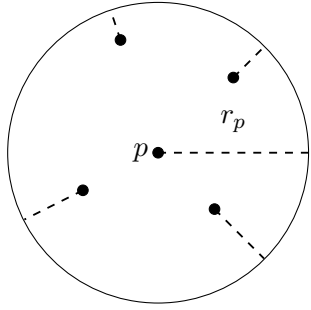
\begin{figure}
\centering
\begin{tikzpicture} 
\draw(0,0) circle (0cm);
\draw (4,0) circle (2cm);
\filldraw 
(4,0) circle (2pt) node[align=left, left] {$p$};
\filldraw 
(5,1) circle (2pt) node[align=left, right] {$ $};
\filldraw 
(3,-.5) circle (2pt) node[align=left, right] {$ $};
\filldraw 
(3.5,1.5) circle (2pt) node[align=left, right] {$ $};
\filldraw 
(4.75,-.75) circle (2pt) node[align=left, right] {$ $};

\draw[dashed, line width= 0.7pt](4,0) --  node [label=above:{$r_p$}] {} (6,0) -- (0:6);
\draw[dashed, line width= 0.7pt](5,1) -- ++(45:0.586);
\draw[dashed, line width= 0.7pt](3,-.5) -- ++(206.57:0.87);
\draw[dashed, line width= 0.7pt](3.5,1.5) -- ++(108.43:0.42);
\draw[dashed, line width= 0.7pt](4.75,-.75) -- ++(-45:0.93);
\end{tikzpicture}
\caption{For each point $p$, we compute a radius $r_p$ such that the dotted lines add to $1$.}
\label{fig:circle_radii}
\end{figure}

We record some properties of $r_p$, some of which were used in previous algorithmic works \cite{MP_alg, sublin_MP}.

\begin{lemma}[Lemma $1$ in \cite{sublin_MP}]\label{lem:intersection}
Every $p \in X$ satisfies $r_p \ge 1/|B(p, r_p) \cap X|$.
\end{lemma}

\begin{proposition}\label{lem:radii_ub}
Let $q \in B(p, r_p) \cap X$. Then $r_q \le 3r_p$.
\end{proposition}

\begin{proof}
Any point $q' \in B(p, r_p) \cap X$ satisfies $\|q-q'\| \le 2r_p$ from the triangle inequality. If we consider the ball $B(q, 3r_p)$ then the sum of the dashed lines in Figure \ref{fig:circle_radii} contributed by points from $B(p, r_p) \cap X$ is at least $r_p$ each. This is because all points $q' \in B(p,r_p) \cap X$ must also be in $B(q, 2r_p)$. The result follows from noting that $|B(p, r_p) \cap X| \ge 1/r_p$ due to Lemma \ref{lem:intersection}.
\end{proof}

\begin{proposition}\label{prob:nearby}
In the optimal solution of \eqref{eq:objective}, every point $p$ must have some $f \in F$ at distance at most $3r_p$. 
\end{proposition}

\begin{proof}
Suppose that a point $p$ does not have a center $f \in F$ within distance $3r_p$. We show in this case that the cost can be reduced. We know from Lemma \ref{lem:intersection} that $|B(p, r_p) \cap X| \ge 1/r_p$. Let $m$ be the number of points in $B(p, r_p) \cap X$ \emph{excluding} $p$. It follows that these points don't have an $f$ within distance $2r_p$. Therefore in total, the contribution of the points in $|B(p, r_p) \cap X|$ to the objective function is at least $2mr_p + 3r_p$. Now if we put a new $f$ at the point $p$, then the $m$ points all have a facility within distance $r_p$ and therefore, the cost of the solution decreases by at least
\[(2mr_p + 3r_p) - (1 + mr_p) = (m+3)r_p -1  = ((m+1)r_p -1)+ 2r_p > 0 \] where the last inequality follows from the fact that $(m+1)r_p \ge 1$ due to Equation \eqref{eq:rdef}. Thus it follows that the optimal solution must have some $f \in F$ that is within distance $3r_p$ of $p$.
\end{proof}

\begin{lemma}\label{lem:const}
There exists constants $c, C > 0$ such that $C \sum_{p \in X} r_p \ge C_n \ge c \sum_{p \in X} r_p$.
\end{lemma}

\begin{proof}
From \cite{MP_alg, sublin_MP} we know that $\sum_{p \in X} r_p$ is a constant factor approximation of $C_n$ \emph{when we restrict} the set $F$ to be a subset of the points $X$. In our case, we want to study a more general version where the set $F$ can come from the entire space. Previous results readily extend to our desired upper bound since not restricting $F$ only decreases the value of the objective function. 

For the lower bound, we denote $C_n'$ as the optimal cost of the objective where $F$ is restricted to points in $X$. Consider the optimal solution for Problem \eqref{eq:objective} (whose objective cost is $C_n$) and denote its set of facilities as $F^*$. For each $f \in F^*$, consider the set of $X$ that it \emph{serves}: for each $f$ we have disjoint subsets $X_f \subseteq X$ such that $f$ is the closest point in $F^*$ to points in $X_f$, breaking ties arbitrarily. Move each $f$ to its closest point in $X_f$. This increases the cost of the objective in \eqref{eq:objective} by at most $\sum_{x\in X} \min_{f \in F^*} \| x-f\|$ since the distance from each point $p \in X_f$ to $f$ increased by at most $\|p-f\|$. Furthermore, we have that this new configuration is a valid solution for the objective where we restrict the set of facilities to come from the points in $X$ and therefore, serves as an upper bound for $C_n'$. Altogether, we have
\[ 2C_n \ge 2|F^*| +  2\sum_{x\in X} \min_{f \in F^*} \| x-f\| \ge  |F^*|  + 2\sum_{x\in X} \min_{f \in F^*} \| x-f\| \ge C_n' \ge c\sum_{p \in X} r_p\]
where the last relation follows from \cite{sublin_MP}. Adjusting the constants gives us our desired bound. 
\end{proof}

Lastly we calculate the expected value of $C_n$ for uniformly random inputs.

\begin{theorem}\label{thm:expectation}
The expected value of the objective \eqref{eq:objective} for i.i.d.\ uniform points in $[0,1]^2$ satisfies
$\E[C_n] = \Theta(n^{2/3})$.
\end{theorem}
\begin{proof}
We know from Lemma \ref{lem:const} that $\sum_{p \in X} r_p$ is a constant factor approximation to the objective given in \eqref{eq:objective}. Therefore, we fix our attention to calculating $\E[r_p]$. Fix a point $p$ and let $r = n^{-1/3}$. The number of points that fall in $B(p, r)$ is distributed as $\text{Bin}(n, cr^2)$ for some constant $c$. By a standard binomial concentration, we know that $|B(p,r) \cap X|  = \Theta(n^{1/3})$ with probability at least $1-e^{-\Theta( n^{1/3} )}$. For example, this follows from Theorems $2.3$ and $2.4$ in \cite{chung2006}.

Conditioning on this event $\mathcal{E}$, we see that from the geometric interpretation of $r_p$ in Figure \ref{fig:circle_radii} that increasing $r$ by $Cn^{-1/3}$ for some sufficiently large constant $C$ will imply $r_p = O(n^{-1/3})$. Thus, 
\[ \E[r_p] \le \E[r_p \mid \mathcal{E}] + \Pr(\mathcal{E}^c) \E[r_p \mid \mathcal{E}^c] = O(n^{-1/3}) + e^{-\Theta(n^{1/3})} = O(n^{-1/3}).\]

For the lower bound, we consider the same approach as above but let $r = c'n^{-1/3}$ for a sufficiently small constant $c'$. In this case, we see that $|B(p,r) \cap X|  \le c''n^{1/3}$ with probability at least  $1-e^{-\Theta( n^{1/3} )}$ for a sufficiently small constant $c''$. Again conditioning on this event $\mathcal{E}$, we see that to make $r_p$ as small as possible, the worst configuration is where all the points in $|B(p,r) \cap X|$ are located at $p$. In that case, we see that $r_p \ge 1/(c''n^{1/3})$. Then we calculate that
\[ \E[r_p] \ge \Pr(\mathcal{E})\E[r_p \mid \mathcal{E}] = \Omega( n^{-1/3} ).\]
The final result follows by linearity of expectations.
\end{proof}
\begin{remark}
Theorem \ref{thm:expectation} is essentially the only place where the uniform distribution assumption and our domain assumption of the unit square in $\mathbb{R}^2$ are used as they allow for an easy calculation of $\E[r_p]$. Most of our concentration arguments in Section \ref{sec:concentration} generalize to arbitrary distributions and arbitrary domains where the appropriate value of $\E[r_p]$ is used.
\end{remark}

\subsection{A Heuristic Derivation of the Concentration Bound}\label{sec:heuristic}
As stated in the introduction, our main result of $C_n$ being concentrated in an interval of length $O(n^{1/6})$ for the case where the input points are i.i.d.\ uniform on the unit square in $\R^2$ can be interpreted as picking a suitable choice of $k$ in Rhee and Talagrand's bound. Indeed, consider the $k$-median problem where $k$ is some parameter specified later. Heuristically, it makes sense to pick the $k$ facilities in a uniform grid of squares of dimension $1/\sqrt{k} \times 1/\sqrt{k}$. In such a case, the distance from any point to its nearest facility is at most $\Theta(1/\sqrt{k})$ and there are $k$ facilities. Thus, the facility location problem objective is $\Theta(n/\sqrt{k}) + k$. Minimizing this as a function of $k$, we see that $k = \Theta(n^{2/3})$. Now Rhee and Talagrand's concentration bound states that the cost of the random $k$-median concentrates on an interval of length $O(\sqrt{n/k})$. `Plugging in' $k = n^{2/3}$ we get an interval of $O(n^{1/6})$ which matches the bound given by Theorem \ref{thm:strong}.

Of course, the above justification is pure heuristics and not rigorous. In addition, Rhee and Talagrand's proof is substantially different than ours. In their work, they exploit the fact that the $k$-median objective is composed of only `local' terms whereas we have a `global' term $|F|$. However, we rely on the geometric properties of the radii $r_p$ outlined above. 

Note that the sum of the radii $r_p$ only serves as a constant factor approximation to the objective value. Therefore, it is \emph{not sufficient} to understand the concentration of the sum of the radii values if we really want to get concentration on the order of $o(\E[C_n])$. Nevertheless, we are able to leverage their properties to provide such a concentration bound for the objective value $C_n$.

\section{Concentration}\label{sec:concentration}
We prove our main concentration inequality in this section. First, we present a suboptimal concentration inequality that follows from Talagrand's inequality. It is interesting to note that an application of this inequality is not sufficient to provide us with the best concentration bound, which is a rare occurrence. Nonetheless, we conjecture that a sharper analysis of our proof using Talagrand's inequality should result in the optimal concentration bound.

We first recall Talagrand's concentration inequality for `non uniform' differences.

\begin{theorem}[Talagrand's Concentration Inequality]\label{thm:talagrand}
Let $f$ be a function on the product space $\Omega = \prod_{i=1}^n \Omega_i$ such that for every $x \in \Omega$, there exists $\alpha_i(x) \ge 0$ with 
\[f(x) \le f(y) + \sum_{i : x_i\ne y_i} \alpha_i(x) \]
for all $y \in \Omega$. Let $M$ denote the median of $f$ and 
\[ c = \sup_{x \in \Omega} \, \sum_{i=1}^n \alpha_i(x)^2. \]
Then,
\[ \Pr(|f-M| \ge t) \le 2e^{- t^2/4c}.\]
\end{theorem}

Using Theorem \ref{thm:talagrand}, we can prove a weaker concentration result that states that $C_n$ is concentrated in an interval of length $n^{1/3}$ in the case of uniform inputs in $[0,1]^2$.

\begin{theorem}[Weak Concentration]\label{thm:weak}
Suppose the points $X = (X_1, \cdots, X_n)$ are chosen independently from some domain $D \subset \R^d$. For each point $p \in X$, define the radius $r_p$ according to \eqref{eq:rdef}. Let $C_n$ denote the cost of the objective function \eqref{eq:objective}. $C_n$ satisfies the concentration inequality
\[ \Pr(|C_n -\textup{Med}(C_n) | \ge t) \le e^{-t^2/s} \]
where $s$ is any upper bound on the quantity $\sum_{p \in X} r_p^2$ for \textbf{any} set of $n$ points chosen from $D$.
In particular, if $D = [0,1]^2$ and the points in $X$ are chosen independently in $[0,1]^2$ (not necessarily from the uniform distribution), we have
\[\Pr(|C_n - \textup{Med}(C_n)| \ge t ) \le e^{-t^2/O(n^{2/3})}\]
where $\textup{Med}(C_n)$ denotes the median value of $C_n$, i.e., $C_n$ is concentrated in an interval of length $n^{1/3}$.
\end{theorem} 

\begin{proof}
Fix an arbitrary collection of points $X = (X_1, \cdots, X_n)$. We define our vector $\alpha$ by letting the $i$th coordinate of $\alpha$ be equal to $Cr_i$ for a suitably large constant $C$. Now given an optimal clustering of a different set of points $Y = (Y_1, \cdots, Y_n)$, we want to extend it to a clustering of $X$ by only using additional `budget' given by  $\sum_{X_i \ne Y_i} \alpha_i(X)$.

Take the set of facilities for $Y$.  Our goal is to show that we can find a facility for every point $p$ in $X \setminus Y$ within distance $O(r_p)$ where the constant in the $O$ doesn't depend on any parameters of the problem. To do this, we first consider the following two cases:
\\

\noindent \textbf{Case 1:}  At least $1/2$ of the points of $B(p, r_p) \cap X$ are in $Y$. 
\\

\noindent Let $q$ be any such point in the intersection. Define $r_q^Y$ be the radius of $q$ calculated according to \eqref{eq:rdef} but using only the points in $Y$. We claim that $r_q^Y = O(r_p)$. To show this, we know from Lemma \ref{lem:intersection} that $|B(p, r_p) \cap X| \ge 1/r_p$ so at least $1/2r_p$ points are in $|B(p, r_p) \cap X \cap Y|$. If we go radius $O(r_p)$ away from $q$, then the sum $\eqref{eq:rdef}$ in $Y$ will be more than $1$, which implies $r_q^Y = O(r_p)$. Thus from Proposition \ref{prob:nearby}, we know that some facility of $Y$ will be within distance $O(r_p)$ from $p$.
\\

\noindent \textbf{Case 2:} At least $1/2$ of the points of $B(p, r_p) \cap X$ are not in $Y$. 
\\

\noindent In this case, we want to find enough points in $B(p,r_p)$ that are in $X$ but not in $Y$ to `pay for a new center' using their radii (that's the budget we are allowed from Theorem \ref{thm:talagrand}). Pick a large constant $C$. We can assume that every $w \in B(p, Cr_p) \cap X$ doesn't fall in case $1$, i.e., the ball $B(w, r_w)$ contains at least $1/2$ of its points from $X$. Indeed, otherwise, $p$ will have a facility in radius $O(r_p)$ from the observation that any $w$ satisfies $r_w = O(r_p)$ from Proposition \ref{lem:radii_ub}. 

Now consider the $w$ in the ball $B(p, r_p) \cap X$ with the smallest radius $r_w$. If $r_w \ge r_p/2$ then we can pay for a new facility from the points in $B(p, r_p) \cap X$ that are not in $Y$ because we know there are at least $1/2r_p$ such points and they all contribute radii $\Omega(r_p)$. If $r_w \le r_p/2$, then we recurse into the ball $B(w, r_w)$. If every $w' \in B(w, r_w) \cap X$ satisfies that $r_w' \ge r_w/2$ then we are again done by the same argument. Otherwise, we again recurse. We know this process ends since we only have $n$ points and when it ends, we are at distance at most $r_p(1+1/2 + 1/4 + \cdots)  \le 2r_p$ away from $p$. Therefore, we can use the entries of $\alpha$ for the points of $B(p, r_p) \cap X$ that are not in $Y$ to pay for a new facility. 

Now to finish our argument for all points, we just repeat the above cases iteratively: We start with a clustering of $Y$ and its facilities. For every point $p \in X$, if it has a facility near $C'r_p$ for a large constant $C'$ then we are done. Otherwise, we consider $B(p, r_p)$ and perform one of the above two cases.

Applying Theorem \ref{thm:talagrand}, we get that $C_n$ satisfies a concentration inequality of the form
\[ \Pr(|C_n -\text{Med}(C_n) | \ge t) \le e^{-t^2/O(\sum_{p \in X} r_p^2)}. \]
Therefore, the value $C_n$ is concentrated in an interval of length $O(\sqrt{\sum_{p \in X} r_p^2})$. To bound this we note that $r_p^2 \le r_p$ since $r_p \le 1$ always since $p$ is included in the sum in Equation \eqref{eq:rdef}. We now claim that $\sum_p r_p = O(n^{2/3})$. This is because $\sum_p r_p$ is constant factor upper bound on facility location cost from Lemma \ref{lem:const} and for every configuration, we can \emph{deterministically} achieve this cost by considering the following construction: place $k$ points in a uniform grid. Then the cost is $n/\sqrt{k} + k$ since each point is within distance $O(1/\sqrt{k})$ from any center. Optimizing for $k$ we get $\sum_{p \in X} r_p = O(n^{2/3})$. 
\end{proof}

We conjecture that the above analysis actually gives us a tighter concentration bound. To show this, we would need a way to control the value of $\sum_{p \in X} r_p^2$ which would depend on the domain $D$ and the how the points in $X$ are drawn.

We now present an argument that gives a much sharper concentration bound using less sophisticated tools.

\begin{theorem}[Strong Concentration]\label{thm:strong} Let $D = [0,1]^2$ and suppose the points in $X$ are i.i.d\ uniform in $[0,1]^2$. Then,
\[\Pr(|C_n - \E[C_n]| \ge t ) \le e^{-t^2/O(n^{1/3})},\] i.e., $C_n$ is concentrated in an interval of length $n^{1/6}$.
\end{theorem}

\begin{proof}
The main ideas of the proof generalize to beyond the uniform in the unit square assumption. We explicitly point out where we assume this in the proof. First, let  $S$ be a set of points in $D$.  We first claim that for any $p \not \in S$, 
\begin{equation}\label{eq:ref1}
    C(S \cup \{p\}) \le C(S) + O(r_p^{S \cup \{p\}})
\end{equation}
where $r_p^{S \cup \{p\}}$ means we calculate the radius \eqref{eq:rdef} with respect to the points $S \cup \{p\}$ and $C(\cdot)$ denotes the facility location problem cost. To show this, we either have $r_p^{S \cup \{p\}} = 1$, in which case we can just put a new facility located at $p$, or otherwise, there must exist some some point $q \in B(p, r_p^{S \cup \{p\}}) \cap S$. That point must have been served in $C(S)$ so by Proposition \ref{prob:nearby}, there must exist a facility near $q$ within distance $3r_q^S$ where we calculate the radius of $q$ with respect to the points in $S$ only. 

Our goal is to show that $r_q^S = O(r_p^{S \cup \{p\}})$. Indeed, if it is the case that $r_q^S \le \|q-p\|$ then clearly $r_q^S = r_q^{S \cup \{p\}}$ since $q$'s radius doesn't change. Else, $p \in B(q, r_q^S)$ in which case $r_q^{S \cup \{p\}}$ is potentially smaller than $r_q^S$. However, considering the geometric interpretation of the radii given in Figure \ref{fig:circle_radii}, we know that the distance contributed by $p$ towards $r_q^{S \cup \{p\}}$ is at most half of the other distances (in other words, the dotted line stemming from $p$ contributes total length at most half to the computation of $r_q^{S \cup \{p\}}$ since there is also a dotted line stemming from $q$). Thus, $r_q^{S \cup \{p\}} \ge r_q^S/2$. Finally from Proposition \ref{lem:radii_ub}, it follows that $r_q^S = O(r_p^{S \cup \{p\}})$. 

We now use our above observation to perform a martingale analysis. Consider the Doob martingale $\Lambda_i = \E[C_n \mid X_1, \ldots, X_i]$ for $1\le i \le n$. We analyze the martingale difference $\Delta_i = \Lambda_i - \Lambda_{i-1}$ which can be written as
\[ \Delta_i =  \E[C_n(X_1, \ldots, X_i, \ldots, X_n )  - C_n(X_1, \ldots, X_i', \ldots, X_n ) \mid X_1, \ldots, X_i] \]
where $X_i'$ is an independent copy of $X_i$. Defining $S = \{X_1, \ldots, X_{i-1}, X_{i+1}, \ldots, X_n \}$, we see that
\[\Delta_i = \E[C_n(S \cup \{X_i\})  - C_n(S \cup \{X_i'\}) \mid X_1, \ldots, X_i] \]
and therefore,
\[|\Delta_i| \le \E[ r_{X_i}^{S \cup \{X_i\}} + r_{X_i'}^{S \cup \{X_i'\}} \mid X_1, \ldots, X_i ] \]
by \eqref{eq:ref1}. Now crucially, we know that the radius defined in \eqref{eq:rdef} \emph{can only decrease} as more points are added. Therefore, we bound each of the expectations above using only the randomness of the remaining $n-i$ points.

Note that everything we have stated so far in the proof is perfectly valid over general domains in $\R^d$ for any $d$ and any choice of distribution that the input points are drawn from, as long as we draw points independently. Now if 
\[ |\Delta_i|^2 \le f(i) \]
for some real valued function $f$, then we immediately arrive at 
\[\Pr(|C_n - \E[C_n]| \ge t ) \le e^{-t^2/\sum_{i=1}^n f(i)}\]
from the Azuma-Hoeffding inequality. The uniform assumption makes the calculation of $f(i)$ particularly tractable. In general, we can calculate $f(i)$ by obtaining a version of Theorem \ref{thm:expectation} for an alternative distribution of choice as we will see shortly.

We now use our assumptions that the domain is $D = [0,1]^2$ and $X$ consists of i.i.d.\ uniform points on the unit square to bound $|\Delta_i|$. This is the main part of the proof where we use the distribution and domain assumptions and the argument which follows can be adapted for other distributional and domain assumptions by finding a suitable bound on $|\Delta_i|$.

From a similar analysis as in Theorem \ref{thm:expectation} (except for a slight caveat that will be addressed in a bit), we know that each of the expectations in our martingale difference can be bounded by $O((n-i)^{-1/3})$ and so it follows that $|\Delta_i| = O((n-i)^{-1/3})$. We now calculate $\sum_i |\Delta_i|^2$. We have that
\begin{equation}\label{eq:var}
    \sum_{i=1}^n (n-i)^{-2/3}  \sim \int_1^n x^{-2/3} \  dx = O(n^{1/3})
\end{equation}
and so by the Azuma-Hoeffding inequality, we get the concentration bound 
\[\Pr( |C_n - \E[C_n]| \ge t )  \le e^{-t^2/O(n^{1/3})}, \]
as desired.

To tie up the loose ends, we note that the upper bound for the expectation given in Theorem \ref{thm:expectation} might not hold if $n-i$ is too small. However, we don't care about this case since we can substitute the deterministic bound $|\Delta_i| = O(1)$ which always holds since the unit square is bounded. In particular, we can use the deterministic bound say when $n-i = O(n^{1/3})$ in which case the `variance' calculation of \eqref{eq:var} still gives us the same asymptotics.
\end{proof}

An interesting open question is if the concentration bound of Theorem \ref{thm:strong} is tight for the uniform unit square case. This is possibly a much harder question but we suspect the answer to be yes. This is because we can show the variance of the quantity $\sum_{p \in X} r_p$ is at least $\Omega(n^{1/3})$ (by performing similar calculations as in Theorem \ref{thm:expectation}) which implies that the quantity $\sum_{p \in X} r_p$ truly fluctuates on an interval of length $\Omega(n^{1/6})$ (the standard deviation). Since this quantity has been very influential in designing algorithms for the facility location problem \cite{MP_alg, sublin_MP}, it hints that our bound is close to the truth. Of course this doesn't formally imply anything for the facility location problem case since $\sum_{p \in X} r_p$ is only a constant factor approximation to the cost. In addition, a $\Omega(n^{1/6})$ fluctuation matches the fluctuation when we `plug in' $k = n^{2/3}$ into the Rhee and Talagrand's bound as done in Section \ref{sec:heuristic}, again hinting that \ref{thm:strong} has the right concentration.

\paragraph{Other Open Problems.}
In this paper we analyzed a well known combinatorial optimization problem under stochastic inputs by borrowing algorithmic ideas that introduce a `local' property. Its plausible that other complicated algorithm problems that allow for local or greedy algorithms can also be analyzed under random inputs. We leave this as an interesting open direction.

\paragraph{Acknowledgements.}
Sandeep Silwal is supported by an NSF Graduate Research Fellowship under Grant No. 1745302, NSF TRIPODS program (award DMS-2022448), and Simons Investigator Award.

\bibliographystyle{abbrv}
\bibliography{main}

\end{document}